\pdfoutput=1
\documentclass{article}
\usepackage{ijcai22}

\usepackage{times}
\usepackage{mathtools}
\usepackage{soul}
\usepackage{url}
\usepackage[hidelinks]{hyperref}
\usepackage[utf8]{inputenc}
\usepackage[small]{caption}
\usepackage{graphicx}
\usepackage{amsmath,amsfonts,amsthm,amssymb}
\usepackage{booktabs}
\usepackage{xspace}
\usepackage{algorithm}
\usepackage{algorithmic}
\usepackage{tikz}
\usetikzlibrary{tikzmark}
\usepackage{array}
\newcolumntype{P}[1]{>{\centering\arraybackslash}p{#1}}
\newcommand{\dyp}{{\mathit{dp}}} 

\usepackage{tabularray}
\usepackage{resizegather}

\newtheorem{theorem}{Theorem}
\newtheorem{lemma}[theorem]{Lemma}

\newtheorem{conjecture}[theorem]{Conjecture}
\newtheorem{proposition}[theorem]{Proposition}

\newtheorem{corollary}[theorem]{Corollary}

\newcommand{\calP}{{\mathcal{P}}}
\newcommand{\calL}{{\mathcal{L}}}
\newcommand{\calI}{{\mathcal{I}}}

\newcommand{\ConeP}{{\sc C1P}\xspace}
\newcommand{\ConePrec}{$\textsc{C1P}^{\textsc{REC}}$\xspace}
\newcommand{\ConePswitch}{$\textsc{C1P}_{\textsc{SW}}$\xspace}
\newcommand{\ConePSwitchrec}{$\textsc{C1P}_{\textsc{SW}}^{\textsc{REC}}$\xspace}
\newcommand{\CConeP}{$\textsc{C1P}_{\medcirc}$\xspace}
\newcommand{\CConePrec}{$\textsc{C1P}_{\medcirc}^{\textsc{REC}}$\xspace}

\newcommand{\two}{two}
\newcommand{\Two}{Two}

\usepackage[T1]{fontenc}
\DeclareSymbolFont{symbolsC}{U}{pxsyc}{m}{n}
\SetSymbolFont{symbolsC}{bold}{U}{pxsyc}{bx}{n}
\DeclareFontSubstitution{U}{pxsyc}{m}{n}
\DeclareMathSymbol{\medcirc}{\mathbin}{symbolsC}{7}

\newcommand{\SIcircle}{$\textsc{SI}_{\medcirc}$\xspace}

\newcommand{\kC}{$k$-{\sc C}\xspace}
\newcommand{\kCrec}{$k$-{\sc C}\textsuperscript{REC}\xspace}

\newcommand{\iverson}[1]{\left[{#1}\right]}
\DeclarePairedDelimiter\floor{\lfloor}{\rfloor}
\DeclarePairedDelimiter\ceil{\lceil}{\rceil}

\usepackage{stmaryrd}

\newcommand{\den}[1]{\left\llbracket{#1}\right\rrbracket}

\title{Voting in \Two-Crossing Elections}

\author{
Andrei Constantinescu\And
Roger Wattenhofer\\
\affiliations
ETH Z{\"u}rich\\
\emails
\{aconstantine, wattenhofer\}@ethz.ch
}

\begin{document}

\maketitle

\begin{abstract}
  We introduce \two-crossing elections as a generalization of single-crossing elections, showing a number of new results. First, we show that \two-crossing elections can be recognized in polynomial time, by reduction to the well-studied consecutive ones problem. We also conjecture that recognizing $k$-crossing elections is NP-complete in general, providing evidence by relating to a problem similar to consecutive ones proven to be hard in the literature. Single-crossing elections exhibit a transitive majority relation, from which many important results follow. On the other hand, we show that the classical Debord-McGarvey theorem can still be proven \two-crossing, implying that any weighted majority tournament is inducible by a \two-crossing election. This shows that many voting rules are NP-hard under \two-crossing elections, including Kemeny and Slater. This is in contrast to the single-crossing case and outlines an important complexity boundary between single- and two-crossing. Subsequently, we show that for \two-crossing elections the Young scores of all candidates can be computed in polynomial time, by formulating a totally unimodular linear program. Finally, we consider the Chamberlin--Courant rule with arbitrary disutilities and show that a winning committee can be computed in polynomial time, using an approach based on dynamic programming. 
\end{abstract}

\section{Introduction}\label{sec:intro}

Many impossibility results in social choice theory disappear if we assume restrictions on the voting preferences. The single-crossing domain is among the most studied restrictions in the literature. Not only does it make many social choice problems tractable, but it is also justifiable practically when placing both voters and candidates on a one-dimensional ``left-right'' spectrum. However, this one-dimensional model is often too restrictive, with real voting preferences hardly ever adhering to the single-crossing model. In this paper we wonder to what degree we can assume a more general model whilst preserving some of the theoretical and algorithmic properties of single-crossingness.

In multi-party democracies, one can often witness situations in which both the far left and the far right oppose against a change brought up by centrist parties. This was prominently happening during the Weimar Republic, but is common also today and known as ``unholy alliance'' or ``The enemy of my enemy is my friend.'' This generalized political model is known as the horseshoe theory; if the far left and the far right are actually closer to each other than to the center, the political line bends into a horseshoe, or even a circle.

An election is single-crossing if voters can be ordered into a list such that, as we sweep from left to right, the relative order between every pair of distinct candidates changes at most once. In this paper, we generalize the single-crossing property to a \two-crossing property, such that the relative order of every pair of candidates is allowed to change at most twice (Figure \ref{2-crossing-example}). As we shall see, this can directly model a horseshoe political system. In contrast to single-crossing preferences, there is evidence that \two-crossing preferences can accurately capture real voting behavior, particularly in economic preferences \cite{double_crossing}. It is worth noting that single-crossing preferences were also first introduced in similar economic settings \cite{mirrlees_income_tax,roberts_single_crossing}. Beyond these practical considerations, we believe that it is worth studying more general domain restrictions, and \two-crossing seems like a natural candidate. 
\begin{figure}
\begin{center}
\begin{tblr}{| *{7}{Q[c, 0.21cm]} |}
\hline
$v_1$ & $v_2$ & $v_3$ & $v_4$ & $v_5$ & $v_6$& $v_7$ \\
\hline
$\tikzmarknode{1_1_2}{1}$ & $\tikzmarknode{2_1_2}{2}$ & $\tikzmarknode{3_1_2}{3}$ & $\tikzmarknode{4_1_2}{3}$ & $\tikzmarknode{5_1_2}{3}$ & $\tikzmarknode{6_1_2}{3}$ & $\tikzmarknode{7_1_2}{1}$ \\
$\tikzmarknode{1_2_2}{2}$ & $\tikzmarknode{2_2_2}{1}$ & $\tikzmarknode{3_2_2}{2}$ & $\tikzmarknode{4_2_2}{4}$ & $\tikzmarknode{5_2_2}{4}$ & $\tikzmarknode{6_2_2}{1}$ & $\tikzmarknode{7_2_2}{3}$ \\
$\tikzmarknode{1_3_2}{3}$ & $\tikzmarknode{2_3_2}{3}$ & $\tikzmarknode{3_3_2}{1}$ & $\tikzmarknode{4_3_2}{2}$ & $\tikzmarknode{5_3_2}{1}$ & $\tikzmarknode{6_3_2}{4}$ & $\tikzmarknode{7_3_2}{4}$ \\
$\tikzmarknode{1_4_2}{4}$ & $\tikzmarknode{2_4_2}{4}$ & $\tikzmarknode{3_4_2}{4}$ & $\tikzmarknode{4_4_2}{1}$ & $\tikzmarknode{5_4_2}{2}$ & $\tikzmarknode{6_4_2}{2}$ & $\tikzmarknode{7_4_2}{2}$ \\
\hline
\end{tblr}\hspace{2cm}
\end{center}
\caption{\Two-crossing profile with $7$ voters and $4$ candidates. Each column, corresponding to a voter, lists alternatives in decreasing order of preference. Voters are given in a \two-crossing order; i.e.~any two of the four colored candidate trajectories cross at most twice.}
\label{2-crossing-example}
\end{figure}
\begin{tikzpicture}[overlay,remember picture, shorten >=-3pt, shorten <= -3pt] 
\draw[line width = 2.5mm, opacity = 0.2, color = blue] (1_1_2.west) -- (1_1_2.east) -- (2_2_2.west) -- (2_2_2.east) -- (3_3_2.west) -- (3_3_2.east) -- (4_4_2.west) -- (4_4_2.east) -- (5_3_2.west) -- (5_3_2.east) -- (6_2_2.west) -- (6_2_2.east) -- (7_1_2.west) -- (7_1_2.east);
\draw[line width = 2.5mm, opacity = 0.2, color = red] (1_2_2.west) -- (1_2_2.east) -- (2_1_2.west) -- (2_1_2.east) -- (3_2_2.west) -- (3_2_2.east) -- (4_3_2.west) -- (4_3_2.east) -- (5_4_2.west) -- (5_4_2.east) -- (6_4_2.west) -- (6_4_2.east) -- (7_4_2.west) -- (7_4_2.east);
\draw[line width = 2.5mm, opacity = 0.2, color = orange] (1_3_2.west) -- (1_3_2.east) -- (2_3_2.west) -- (2_3_2.east) -- (3_1_2.west) -- (3_1_2.east) -- (4_1_2.west) -- (4_1_2.east) -- (5_1_2.west) -- (5_1_2.east) -- (6_1_2.west) -- (6_1_2.east) -- (7_2_2.west) -- (7_2_2.east);
\draw[line width = 2.5mm, opacity = 0.2] (1_4_2.west) -- (1_4_2.east) -- (2_4_2.west) -- (2_4_2.east) -- (3_4_2.west) -- (3_4_2.east) -- (4_2_2.west) -- (4_2_2.east) -- (5_2_2.west) -- (5_2_2.east) -- (6_3_2.west) -- (6_3_2.east) -- (7_3_2.west) -- (7_3_2.east);
\end{tikzpicture}

Determining the winners of an election is perhaps the oldest and most fundamental problem in the field. In most single-winner voting systems, finding the winners is straightforward, with a few notable exceptions, for which the problem is at least NP-hard:~Young  \cite{rothe_young_completeness,brandt_weak_young_completeness,young_score_np_complete}, Dodgson \cite{kemeny_score_np_hard,hemaspaandra_theta2_hard_dodgson} and Kemeny \cite{kemeny_score_np_hard,kemeny_theta_2_p_hard}. For single-crossing preferences, all the aforementioned admit polynomial time algorithms, hinging essentially on the existence of (weak) Condorcet winners, which are easy to compute and, at least for an odd number of voters, actually give the winner straight away. For \two-crossing elections, the situation is more interesting, as we shall see.

When it comes to multi-winner rules, there are many prominent examples which are NP-hard, including one of the most studied, the Chamberlin--Courant rule for proportional representation, which is hard even when voters' dissatisfaction values follow very simple patterns \cite{prz08,lb11}. On the other hand, computing a winning committee is easy for single-crossing preferences \cite{cc_sc_poly_time_edith,constantinescu_elkind_2021}, but NP-hard for three-crossing preferences \cite{cc_np_hard_3_crossing}. Studying the \two-crossing case is required to close this gap.

\paragraph{Our Contribution.} We study \two-crossing elections from an axiomatic and algorithmic point of view. Axiomatically, we show that they are equally expressive to unrestricted elections in terms of the (weighted) majority tournament:~all weighted majority tournaments with same-parity weights are inducible by a \two-crossing profile. Consequently, Slater, Banks, Minimal Extending Set, Tournament Equilibrium Set, Kemeny and Ranked Pairs are all NP-hard under \two-crossing elections.
Algorithmically, we show how recognition can be achieved in polynomial time and study the winner determination problem for Young's and the Chamberlin--Courant rule, in both cases providing polynomial time algorithms. We also ask the question whether recognizing $k$-crossingness is NP-complete in general. We conjecture that the answer is yes for $k \geq 4$ and give evidence by relating to the work of \cite{goldberg}.

\section{Preliminaries}\label{sec:prelim}
Given integer $n$, we write $[n]$ for the set $\{1, \dots, n\}$;
given two integers $\ell \leq r$, we write $[\ell:r]$ to denote the set
$\{\ell, \dots, r\}$. For a fixed $n$ and $1 \leq \ell < r \leq n$ we write $[r : \ell]$ for the set $[1 : \ell] \cup [r : n]$. A subset $\calI \subseteq [n]$ is an {\em interval} if $\calI = [\ell : r]$ for some $1 \leq \ell \leq r \leq n$ and a {\em circular interval} if we also allow for $\ell > r$ in the preceding condition. For a statement $S$, we write $\iverson{S}$ for the Iverson bracket:~$[S] = 1$ if $S$ holds, and $0$ otherwise.

We consider a setting where voters $V = [n]$ express their preferences over a set of candidates (or alternatives) $C=[m]$. Voters rank candidates from best to worst,
so that the preferences of a voter $v$ are given by a linear order $\succ_v$:~given two distinct candidates $c, c'\in C$ we write $c\succ_v c'$
when $v$ prefers $c$ to $c'$. The list of all voters' preferences is denoted by $\calP = (\succ_v)_{v\in V}$, and is referred to as a {\em preference profile}.
We use $\calL(C)$ to denote the set of all linear orders over $C$, such that $\calP \in \calL(C)^n$.

\paragraph{Pairwise Majority and Young's Rule.}

Given a profile $\calP$ over candidate set $C$ and two candidates $c, c' \in C$ we write $n_{c, c'} = |\{i \in V : c \succ_i c'\}|$ for the number of voters who prefer $c$ to $c'$; the so-called \emph{majority margin} is then $m_{c, c'} = n_{c, c'} - n_{c', c}$.
Consider orienting the edges of the complete graph on $C$ such that $c \rightarrow c'$ iff $m_{c, c'} > 0$; leaving out edges with $m_{c, c'} = 0$. This construction is known as the \emph{majority tournament} of election $\calP$. If each edge $c \rightarrow c'$ is also given weight $m_{c, c'}$, then the resulting construction is known as the \emph{weighted majority tournament} of $\calP$. If there is a candidate $c \in C$ such that $m_{c, c'} > 0$ for all $c' \neq c$, then $c$ is called the \emph{strong Condorcet winner}.
If the previous condition is weakened to $m_{c, c'} \geq 0$, then $c$ is only a \emph{weak Condorcet winner}, which there can be multiple of.
For profile $\calP$, the \emph{strong (weak) Young score} of a candidate $c \in C$ is the minimum number of voters that need to be removed from $\calP$ such that $c$ becomes a strong (weak) Condorcet winner; in some cases this score will be infinite.
In the strong (weak) \emph{Young's rule} candidates with the smallest strong (weak) Young score are declared the winners. Traditionally, ``Condorcet winner'' means the strong variant, while ``Young's rule'' means the weak variant \cite{handbook_young_chapter}.

\paragraph{Tournament Solutions.}
Under some voting rules, knowledge of the (weighted) majority tournament is enough to compute the winners; we call such rules \emph{(weighted) tournament solutions}. Some tournament solutions allow for polynomial time winner determination, the most natural example being Copeland. However, for a significant number of tournament solutions it is NP-hard to determine the winners --- unweighted examples:~Slater, Banks, Minimal Extending Set, Tournament Equilibrium Set; weighted examples:~Kemeny, Ranked Pairs. See \cite{handbook_unweighted_tournaments_chapter,handbook_kemeny_chapter} for a survey on the topic.

In what follows, we assume that we are also given a {\em misrepresentation function}
$\rho: V\times C\to {\mathbb Q}$; $\rho$ is {\em consistent} with $\calP$
if $c\succ_v c'$ implies $\rho(v, c)\le \rho(v, c')$ for all $v\in V$ and $c, c'\in C$.
Intuitively, the value $\rho(v, c)$ indicates to what extent
candidate $c$ misrepresents voter $v$. 
A common example of a misrepresentation function is the {\em Borda function}
$\rho_B$ given by $\rho_B(v, c)=|\{c'\in C: c'\succ_v c\}|$;
this function assigns value $0$ to a voter's top choice, 
value $1$ to his second choice, and value $m-1$ to his last choice. We assume that operations on the values of $\rho(v, c)$ (e.g.~addition) can be performed in unit time; this assumption is realistic as the values of $\rho$ are usually small integers; e.g.~$\rho = \rho_B$.

\paragraph{The Chamberlin--Courant Rule.}
A {\em multiwinner voting rule} maps a profile $\calP$ over a candidate set 
$C$, and a positive integer $k \leq |C|$, to a non-empty collection 
of subsets of $C$ of size at most $k$;\footnote{Usually exactly $k$, but in our case the difference is immaterial.} the elements of this collection are called 
the {\em winning committees}. An {\em assignment function} is a mapping $w: V\to C$; for each $V'\subseteq V$
we write $w(V')=\{w(v): v\in V'\}$ for the image of $V'$ under $w$. If $|w(V)|\le k$, 
then $w$ is called a {\em $k$-assignment}. Given 
a misrepresentation function $\rho$ and
a profile $\calP=(\succ_v)_{v\in V}$, 
the {\em total dissatisfaction of voters in $V$ under a $k$-assignment $w$}
is given by $\Phi_\rho(\calP, w) = \sum_{v\in V}\rho(v, w(v))$. Intuitively, 
$w(v)$ is the representative of voter $v$ in the committee $w(V)$, 
and $\Phi_\rho(\calP, w)$ measures to what extent the voters are dissatisfied
with their representatives. An {\em optimal $k$-assignment}
for $\rho$ and $\calP$ is a $k$-assignment that minimizes $\Phi_\rho(\calP, w)$
among all $k$-assignments for $\calP$.
The {\em Chamberlin--Courant multiwinner voting rule} \cite{cc83,mw-chapter} maps each triple $(\calP, \rho, k)$ consisting of a preference profile $\calP=(\succ_v)_{v\in V}$
over a candidate set $C$, a misrepresentation function $\rho: V\times C\to{\mathbb Q}$ that is consistent with $\calP$, and a positive integer $k\le |C|$, to all sets $W$ such that $W=w(V)$ for some $k$-assignment $w$ that is optimal
for $\rho$ and $\calP$, constituting the winning committees.\footnote{Most literature does not make $\rho$ an input argument, each choice for $\rho$ making for a different rule. The implications of this distinction are minor for our work. In practice, $\rho$ can, for instance, take the form of an $n \times m$ matrix.}
It is NP-hard to determine whether a $k$-assignment of dissatisfaction $\Phi_\rho(\calP, w) \leq B$ exists for some input parameter $B$, even in the special cases where $\rho(v, c) \in \{0, 1\}$ \cite{prz08}, or if $\rho$ is the Borda function \cite{lb11}.

\paragraph{$k$-Crossing Preferences.}
A profile $\calP=(\succ_v)_{v\in V}$ over $C$ is {\em $k$-crossing}
if there is a permutation $(\sigma_i)_{i \in V}$ of $V$ such that for every pair of distinct
candidates $(c, c')\in C^2$ the number of indices $i \in [n - 1]$ such that $\iverson{c \succ_{\sigma_i} c'} \neq \iverson{c \succ_{\sigma_{i + 1}} c'}$ is at most $k$. That is, if we order the voters in $V$ according to $\sigma$ and traverse the list of voters from left to right, each pair of candidates `crosses' at most $k$ times. In this case we write that $\calP$ is \kC, with respect to $\sigma$.
A profile $\calP$ is {\em single-interval on a circle} if there is a permutation $(\sigma_i)_{i \in V}$ such that for all pairs of distinct candidates $(c, c') \in C^2$ the set $\{i \in V : c \succ_{\sigma_i} c'\}$ is a circular interval over $[n]$. If so, we write that $\calP$ is \SIcircle on the circle induced by $\sigma$. For $\Phi \in \{\text{\kC}, \text{\SIcircle}\}$ the \emph{recognition} problem $\Phi^\textsc{REC}$ asks:~given a preference profile $\calP$, is there a permutation $\sigma$ such that $\calP$ is $\calL$ with respect to $\sigma$? If yes, one is also interested in finding such a witnessing $\sigma$.

\begin{lemma} \label{lemma:2-cross-circle} For any permutation $(\sigma_i)_{i \in V}$, profile $\calP$ is \two-crossing with respect to $\sigma$ iff $\calP$ is single-interval on the circle induced by $\sigma$. Consequently, $\calP$ is \two-crossing iff $\calP$ is \SIcircle.
\end{lemma}
\begin{proof} Consider two candidates $c \neq c'$. If $\{i \in V : c \succ_{\sigma_i} c'\}$ is a circular interval, then there are at most two crosses, one at each interval end. Conversely, assume there are two crosses (the cases with one/zero crosses are similar) at positions $1\leq i_1 < i_2 < n$ and that, without loss of generality, $\sigma_1$ prefers $c$ to $c'$. Then, voters in $[i_1 + 1 : i_2]$ prefer $c'$ to $c$ and voters in $[i_2 + 1 : i_1]$ prefer $c$ to $c'$, both being circular intervals.  
\end{proof}

\begin{lemma} Consider a horseshoe political system:~voters and candidates are assigned points on the unit circle, voters rank candidates in increasing order of circle arc-length distance from their assigned point, assume no ties. Then, voters' preferences are \two-crossing.
\end{lemma}

\begin{proof}
Without loss of generality, assume the voters $V = [n]$ are ordered in increasing order of positive arc-length distance from some fixed point on the circle. Consider any pair of distinct candidates $(c, c') \in C^2$ and let $A_c$ and $A_{c'}$ be their assigned points on the circle. Then, voters with assigned points on one side of the perpendicular bisector of $A_cA_c'$ will prefer $c$ to $c'$ (and vice-versa for the other side). By the assumption on the ordering of $V$, these voters preferring $c$ to $c'$ will form a circular interval over $[n]$. Consequently, the voters' preferences are \SIcircle, so they are \two-crossing by Lemma \ref{lemma:2-cross-circle}.
\end{proof}

\paragraph{The Consecutive Ones Problem.} An $n \times m$ binary matrix $M$ has the {\em consecutive ones property} if its rows can be permuted such that in each column ones form a single continuous run. In this case, we say that $M$ is \ConeP. Generalizing, $M$ is $k$-\ConeP if its rows can be permuted such that in each column ones form at most $k$ continuous runs. Similarly, $M$ has the {\em circular consecutive ones property} if its rows can be permuted such that in each column ones form a single continuous run if we allow loop-around; we use \CConeP to denote this property. For $\Phi \in \{\text{\ConeP}, k\text{-\ConeP}, \text{\CConeP}\}$ we use $\Phi^\textsc{REC}$ to denote the corresponding recognition problems. It is instructive to note the following facts tying \ConeP and \CConeP.

\begin{lemma}[{\cite[Theorem 1]{tucker}}]\label{lemma_circle_is_line} From a binary matrix $M$ construct $M^c$ by flipping ones and zeros on those columns with a one in the first row of $M$.
Then, any permutation of rows $\sigma$ fixing the first row witnesses that $M$ is \CConeP iff $\sigma$ witnesses that $M^c$ is \ConeP.
\end{lemma}

\begin{proof}
Flipping a column of $M$ does not change whether it is \CConeP, so $M$ is \CConeP iff $M^c$ is \CConeP. Moreover, this transformation does not affect the set of witnessing permutations. By construction, $M^c$ has only zeroes on its first row, and we lose no generality by forcing a given row to be first for \CConeP, hence reducing to the non-circular variant \ConeP.
\end{proof}

\begin{corollary}\label{coro_lemma_circle_is_line}
\CConePrec and \ConePrec are equivalent (interreducible with respect to complexity-preserving reductions). This also extends to finding witnessing permutations.
\end{corollary}

\begin{proof}
Whenever a permutation witnessing that some matrix $M$ is \CConeP is cyclically permuted, the result also witnesses the property. Therefore, restricting to permutations fixing the first row loses no generality, so checking whether $M$ is \CConeP reduces to checking whether $M^c$ is \ConeP, by Lemma \ref{lemma_circle_is_line}. The converse also holds:~to check whether $M$ is \ConeP, add an extra row of zeroes and check for \CConeP.
\end{proof}

\noindent Recognizing \ConeP/\CConeP matrices and finding witnessing permutations can both be achieved in time linear in the size of the matrix \cite{booth_lueker}. A simpler $O(nm^2)$ algorithm was given by \cite{fulkerson_gross}. For $k \geq 2$, $k$-\ConePrec is NP-complete \cite{goldberg}.

\section{Recognizing $k$-Crossing Profiles}

The problem of recognizing single-crossing elections admits a number of polynomial time algorithms \cite{elkind_reco_sc,single_crossing_characterization}, but, to the best of our knowledge, recognizing $k$-crossing profiles for $k > 1$ has not been studied. We prove that recognizing \two-crossing profiles reduces to recognizing \ConeP matrices, which is tractable in poly-time. We leave the problem open for $k > 2$. However, for $k \geq 4$ we conjecture that recognition is NP-complete.
\begin{conjecture} \kCrec is NP-complete for $k \geq 4$.
\end{conjecture}

\subsection{Polynomial Time Recognition for $k = 2$}

Given a profile $\calP$ over candidate set $C$, let $M_\calP$ be a binary matrix with $n$ rows and $m(m - 1)$ columns:~one row for each voter $v \in V$ and one column for each pair of distinct candidates $(c, c') \in C^2$; such that $M_\calP[v, (c, c')] = \iverson{c \succ_{v} c'}$.

\begin{lemma} \label{lemma:sc-circle-vs-matrix}For any permutation $(\sigma_i)_{i \in V}$, $\calP$ is single-interval on the circle induced by $\sigma$ iff $M_\calP$ is \CConeP with respect to $\sigma$. Consequently, $\calP$ is single-interval on a circle iff $M_\calP$ is \CConeP.
\end{lemma}
\begin{proof} Consider two candidates $c \neq c'$. The requirement on $(c, c')$ for \SIcircle on the circle induced by $\sigma$ to hold is that the set $\{i \in V : c \succ_{\sigma_i} c'\}$ is a circular interval over $[n]$. This is equivalent to saying that its indicator function $\iverson{c \succ_{\sigma_i} c'}$ has all ones grouped into a single circular run, which is, by definition of $M_\calP$, the same as saying that column $(c, c')$ of $M_{\calP}$ reordered by $\sigma$ has all ones in a single circular interval. Quantifying over all $c \neq c'$ gives the conclusion.
\end{proof}

\begin{theorem} \label{theorem_reco} Given a preference profile $\calP$, deciding whether it is \two-crossing and, if affirmative, finding a witnessing permutation can be done in time $O(nm^2)$. 
\end{theorem}

We note that the construction for $M_\calP$ was first suggested by \cite{single_crossing_characterization}, but their result is somewhat different:~they show that $\calP$ is single-crossing iff $M_\calP$ is \ConeP , while we show that $\calP$ is \two-crossing iff $M_\calP$ is \CConeP (i.e.~$M_\calP^c$ is \ConeP). The relationship between consecutive ones and $k$-crossingness seems to go deeper than this. Say $\calP$ is $k$-interval on a circle if voters can be rearranged into a circle such that for all pairs of distinct candidates $(c, c') \in C^2$ the set $\{i \in V : c \succ_{i} c'\}$ is a union of at most $k$ circular intervals. Similarly, say a matrix is $k$-\CConeP if its rows can be rearranged such that each column has at most $k$ circular runs of ones. With these conventions, one can follow the proof of Lemma \ref{lemma:2-cross-circle} to also show that $\calP$ is $2k$-crossing iff $\calP$ is $k$-interval on a circle. Moreover, one can follow the proof of Lemma \ref{lemma:sc-circle-vs-matrix} to show that $\calP$ is $k$-interval on a circle iff $M_\calP$ is $k$-\CConeP.

\subsection{Is Recognition NP-Complete for $k > 2$\,?}

In this section we make progress towards showing that recognizing $k$-crossing elections is NP-complete starting with some value of $k$. Recall that for $k \leq 2$ the problem is tractable, while for $k > 2$ it is open. Our work is inspired by the following result.

\begin{theorem}[{\cite[Theorem 5.2]{goldberg}}]\label{goldberg_np_complete}
$k\text{-\ConePrec}$ is NP-complete for $k \geq 2$.
\end{theorem}

\noindent We say a binary matrix $(M_{ij})_{i \in [n], j \in [m]}$ has property $k$-\ConePswitch if there is a permutation $\sigma$ of the rows of $M$ such that for all columns $j \in [m]$ there are at most $k$ positions $i \in [n - 1]$ such that $M_{\sigma_ij} \neq M_{\sigma_{i + 1}j}$. In other words, on each column of the permuted matrix there are at most $k$ ``switches'' from zero to one or vice-versa. We write $k$-\ConePSwitchrec for the associated recognition problem.
For building intuition, note that any $k$-\ConeP matrix is also $2k$-\ConePswitch, but not the other way around. Similarly, any $k$-\ConePswitch matrix is also ${\left(1 + \floor*{k/2}\right)}$-\ConeP, but not the other way around. Given these, one can see an informal ``almost equality'' relation between $k$-\ConeP and {$2k$-\ConePswitch}, modulo $\pm1$ variations in $k$, so it is natural to also expect that {$k$-\ConePSwitchrec} is NP-complete for $k$ above a certain threshold, by reduction from some $k'$-\ConePrec. Despite this similarity and our best efforts, we leave finding such a reduction open for now.

\begin{conjecture} \label{conj} $k$-\ConePSwitchrec is NP-complete for $k \geq 4$.
\end{conjecture}

\noindent However, we give the other piece of the puzzle:~a reduction from $k$-\ConePSwitchrec to \kCrec, so that only a proof of Conjecture \ref{conj} would suffice to show NP-completeness for \kCrec. For the reduction, consider a binary matrix ${M = (M_{ij})_{i\in [n], j \in [m]}}$. From $M$ one can construct a preference profile $\calP_M$, as follows:~voters $V = [n]$ are in one-to-one correspondence with the rows; for each column $j \in [m]$ introduce two candidates $c^1_j$, $c^2_j$ such that $C = \{c^\ell_j : \ell \in [2], j \in [m]\}$. Voters rank candidates corresponding to different columns consistently, and in a way which does not depend on $M$ itself; in particular, we set $c^{\ell_1}_{j_1} \succ_{v} c^{\ell_2}_{j_2}$ for all $v \in V$, $\ell_1, \ell_2 \in [2]$ and $1 \leq j_1 < j_2 \leq m$. The only freedom we are now left with in profile $\calP_M = (\succ_v)_{v \in V}$ is in specifying for each $i \in [n]$ and $j \in [m]$ whether $c^1_j \succ_{i} c^2_j$ or vice-versa. To settle this, we arrange so that $\iverson{c^1_j \succ_{i} c^2_j} = M_{ij}$. 

\begin{lemma} \label{lemma_reduction}\label{coro:kswitch-kc} For any permutation of rows $\sigma$, $M$ is $k$-\ConePswitch with respect to $\sigma$ iff $\calP_M$ is \kC with respect to $\sigma$. Consequently, $M$ is $k$-\ConePswitch iff $\calP_M$ is \kC.
\end{lemma}
\begin{proof} Note that, irrespective of $\sigma$, by construction no pair of candidates $(c_{j_1}^{\ell_1}$, $c_{j_2}^{\ell_2})$ with $j_1 \neq j_2$ can generate any crosses, so we can, without loss of generality, restrict the definition of \kC to consider only pairs of the form ${(c_{j}^1, c_{j}^2)}$ for $j \in [m]$. Consider some fixed $j \in [m]$, then the number of crosses induced by $(c_{j}^1$, $c_{j}^2)$ for $\sigma$ is given by ${|\{i \in [n - 1] : \iverson{c_{j}^1 \succ_{\sigma_i}c_{j}^2} \neq \iverson{c_{j}^1 \succ_{\sigma_{i + 1}}c_{j}^2}\}|}$, which is, by construction, the same as ${|\{i \in [n - 1] : M_{\sigma_ij} \neq M_{\sigma_{i + 1}j}\}|}$, the latter being the number of switches on column $j$ of $M$ when reordered by $\sigma$. Since the number of crosses induced by $(c_j^1, c_j^2)$ and the number of switches on column $j$ match, we get the conclusion by quantifying over all $j \in [m]$ whilst recalling that $j \in [m]$ are in bijection both with columns of $M$ and with pairs of candidates $(c_{j}^1$, $c_{j}^2)$ in $\calP_M$.
\end{proof}

\begin{theorem} \label{first_conditional_theorem} Assuming Conjecture \ref{conj}, \kCrec is NP-complete for $k \geq 4$. 
\end{theorem}
\begin{proof} \kCrec is clearly in NP. To show NP-hardness we reduce from $k$-\ConePSwitchrec. Given an instance $M$ of $k$-\ConePSwitchrec, the reduction constructs $\calP_M$ in polynomial time. By Lemma \ref{coro:kswitch-kc}, $M$ is $k$-\ConePswitch iff $\calP_M$ is \kC, giving the conclusion.
\end{proof}

For even values of $k$, one can also consider a second proof strategy, based on the following conjecture.

\begin{conjecture} \label{conj_2} $k$-\CConePrec is NP-complete for $k \geq 2$.
\end{conjecture}

\noindent Observe that this is weaker than Conjecture \ref{conj}, because of the following fact.

\begin{lemma} For $k \geq 1$ it holds that $k\text{-\CConeP} = 2k\text{-\ConePswitch}$.
\end{lemma}
\begin{proof} Double-inclusion following the definitions.%
\end{proof}

Note that we choose not to make any claims about the hardness of $3$-\ConePSwitchrec, which unfortunately means no conjecture about the hardness of recognizing three-crossing elections. For all we know, this might be polynomial, or it might just as well be NP-hard.

\section{Weighted Majority Tournaments}

Single-crossing elections are attractive from an axiomatic standpoint:~the majority relation is transitive and a Condorcet winner exists, for odd $n$. The Condorcet Paradox profile consists of $3$ voters, so it can be easily seen as \two-crossing. Therefore, Condorcet winners might not exist under \two-crossing. This begs the question, can we guarantee anything about the weighted majority tournament of a \two-crossing election? The following result, similar to \cite{dominik_spoc} for elections single peaked on a circle, answers negatively.

\begin{theorem}[Debord-McGarvey for \Two-Crossing Elections] \label{mcgarvey}Any weighted majority tournament with weights of the same parity%
\footnote{Nonexistent edges are considered to have weight $0$.}
is inducible by a \two-crossing election. If $W$ denotes the maximum weight of an edge, then $O(m^2W)$ voters suffice. 
\end{theorem}
\begin{proof} For some number of candidates $m$, consider the following ``Double-BubbleSort'' construction of a \two-crossing profile. There will be $m(m - 1) + 1$ voters; voter $1$ ranks ${1 \succ 2 \succ \ldots \succ m}$, which we more succinctly represent as the permutation $1 2 3\ldots m$. Voter $2$ ranks $2 1 3\ldots m$, voter $3$ ranks $2 3 1 \ldots m$, and so on, voter $m$ ranks $2 3 \ldots m 1$. In essence, one swap at a time, candidate $1$ went from best to worst. For the following $m - 1$ voters, candidate $2$ will go from best to second worst, one position at a time:~$3 2 4 \ldots m 1$, $3 4 2 \ldots m 1$, \ldots, $3 4 \ldots m 2 1$. In the following rounds, we similarly bring from front to back candidates $3, 4, \ldots m - 1$, each taking multiple swaps to reach their final position. Note that the swaps we do are precisely those done by a BubbleSort algorithm sorting in descending order. This construction is illustrated for $m = 4$ by the first $7$ voters in Figure \ref{2-crossing-mcgarvey}. To complete the profile, we need $m(m - 1)/2$ additional voters; consecutive voters will, once again, differ by a single adjacent swap. In particular, we go from $m\ldots21$ back to $12\ldots m$ by following the swaps of a BubbleSort algorithm sorting in ascending order, but iterating over the permutation in reverse order. Figure \ref{2-crossing-mcgarvey}, voters $v_7$ to $v_{13}$,  demonstrates this process. Note that the majority margins satisfy $m_{c, c'} = 1$ for $c < c'$. 

Why is this profile interesting? For any two candidates $c \neq c'$ there are voters $v_i$, $v_j$ with preference permutations $\tau_i = Acc'B$ and $\tau_j = \overline{B}cc'\overline{A}$; where $\overline{X}$ denotes the reverse of permutation $X$. This means that, if we add one additional copy of voters $v_i$, $v_j$ to our profile, then all majority margins stay unchanged, except for $m_{c, c'}$ and $m_{c', c}$, which increase/decrease by $2$. Since $c$, $c'$ were arbitrary, this means that we can increase/decrease any majority margin by $2$ without affecting the others,%
\footnote{Subject to the implicit constraint that $m_{c, c'} + m_{c', c} = 0$.}
by adding two additional voters. If weights are odd, this can be done repeatedly until our profile has margins agreeing to the weights in the tournament. Fixing the margin for a weight $w \leq W$ 
requires $O(w)$ additions of voters, so overall we need at most $O(m^2W)$ voters. The case of even weights is similar:~add an additional $m\ldots 2 1$ voter to our original profile, and then proceed analogously.
\end{proof}

The previous result shows that, essentially, the computational complexity of all known (weighted) tournament solutions is unchanged from the general case by restricting to \two-crossing elections. Thus, we have the following.

\begin{corollary}
Determining the winners for Slater, Banks, Minimal Extending Set, Tournament Equilibrium Set, Kemeny and Ranked Pairs is NP-hard under \two-crossing elections.
\end{corollary}

\begin{figure}
\begin{center}
\begin{tblr}{| *{6}{Q[c, 0.18cm]} | Q[c, 0.21cm] | *{2}{Q[c, 0.17cm]} *{4}{Q[c, 0.32cm]} |}
\hline
$v_1$ & $v_2$ & $v_3$ & $v_4$ & $v_5$ & $v_6$& $v_7$& $v_8$& $v_9$& $v_{10}$ & $v_{11}$ & $v_{12}$ & $v_{13}$\\
\hline
$\tikzmarknode{1_1}{1}$ & $\tikzmarknode{2_1}{2}$ & $\tikzmarknode{3_1}{2}$ & $\tikzmarknode{4_1}{2}$ & $\tikzmarknode{5_1}{3}$ & $\tikzmarknode{6_1}{3}$ & $\tikzmarknode{7_1}{4}$ & $\tikzmarknode{8_1}{4}$ & $\tikzmarknode{9_1}{4}$ & $\tikzmarknode{10_1}{1}$ & $\tikzmarknode{11_1}{1}$ & $\tikzmarknode{12_1}{1}$ & $\tikzmarknode{13_1}{1}$ \\
$\tikzmarknode{1_2}{2}$ & $\tikzmarknode{2_2}{1}$ & $\tikzmarknode{3_2}{3}$ & $\tikzmarknode{4_2}{3}$ & $\tikzmarknode{5_2}{2}$ & $\tikzmarknode{6_2}{4}$ & $\tikzmarknode{7_2}{3}$ & $\tikzmarknode{8_2}{3}$ & $\tikzmarknode{9_2}{1}$ & $\tikzmarknode{10_2}{4}$ & $\tikzmarknode{11_2}{4}$ & $\tikzmarknode{12_2}{2}$ & $\tikzmarknode{13_2}{2}$ \\
$\tikzmarknode{1_3}{3}$ & $\tikzmarknode{2_3}{3}$ & $\tikzmarknode{3_3}{1}$ & $\tikzmarknode{4_3}{4}$ & $\tikzmarknode{5_3}{4}$ & $\tikzmarknode{6_3}{2}$ & $\tikzmarknode{7_3}{2}$ & $\tikzmarknode{8_3}{1}$ & $\tikzmarknode{9_3}{3}$ & $\tikzmarknode{10_3}{3}$ & $\tikzmarknode{11_3}{2}$ & $\tikzmarknode{12_3}{4}$ & $\tikzmarknode{13_3}{3}$ \\
$\tikzmarknode{1_4}{4}$ & $\tikzmarknode{2_4}{4}$ & $\tikzmarknode{3_4}{4}$ & $\tikzmarknode{4_4}{1}$ & $\tikzmarknode{5_4}{1}$ & $\tikzmarknode{6_4}{1}$ & $\tikzmarknode{7_4}{1}$ & $\tikzmarknode{8_4}{2}$ & $\tikzmarknode{9_4}{2}$ & $\tikzmarknode{10_4}{2}$ & $\tikzmarknode{11_4}{3}$ & $\tikzmarknode{12_4}{3}$ & $\tikzmarknode{13_4}{4}$ \\
\hline
\end{tblr}\hspace{2cm}
\end{center}
\caption{Construction in the proof of Theorem \ref{mcgarvey} for $m = 4$ candidates. Voters are arranged in a \two-crossing order.
}
\label{2-crossing-mcgarvey}
\end{figure}

\begin{tikzpicture}[overlay,remember picture, shorten >=-3pt, shorten <= -3pt] 
\draw[line width = 2.5mm, opacity = 0.2, color = blue] (1_1.west) -- (1_1.east) -- (2_2.west) -- (2_2.east) -- (3_3.west) -- (3_3.east) -- (4_4.west) -- (4_4.east) -- (5_4.west) -- (5_4.east) -- (6_4.west) -- (6_4.east) -- (7_4.west) -- (7_4.east) -- (8_3.west) -- (8_3.east) -- (9_2.west) -- (9_2.east) -- (10_1.west) -- (10_1.east) -- (11_1.west) -- (11_1.east) -- (12_1.west) -- (12_1.east) -- (13_1.west) -- (13_1.east);
\draw[line width = 2.5mm, opacity = 0.2, color = orange] (1_2.west) -- (1_2.east) -- (2_1.west) -- (2_1.east) -- (3_1.west) -- (3_1.east) -- (4_1.west) -- (4_1.east) -- (5_2.west) -- (5_2.east) -- (6_3.west) -- (6_3.east) -- (7_3.west) -- (7_3.east) -- (8_4.west) -- (8_4.east) -- (9_4.west) -- (9_4.east) -- (10_4.west) -- (10_4.east) -- (11_3.west) -- (11_3.east) -- (12_2.west) -- (12_2.east) -- (13_2.west) -- (13_2.east);
\draw[line width = 2.5mm, opacity = 0.2, color = red] (1_3.west) -- (1_3.east) -- (2_3.west) -- (2_3.east) -- (3_2.west) -- (3_2.east) -- (4_2.west) -- (4_2.east) -- (5_1.west) -- (5_1.east) -- (6_1.west) -- (6_1.east) -- (7_2.west) -- (7_2.east) -- (8_2.west) -- (8_2.east) -- (9_3.west) -- (9_3.east) -- (10_3.west) -- (10_3.east) -- (11_4.west) -- (11_4.east) -- (12_4.west) -- (12_4.east) -- (13_3.west) -- (13_3.east);
\draw[line width = 2.5mm, opacity = 0.2] (1_4.west) -- (1_4.east) -- (2_4.west) -- (2_4.east) -- (3_4.west) -- (3_4.east) -- (4_3.west) -- (4_3.east) -- (5_3.west) -- (5_3.east) -- (6_2.west) -- (6_2.east) -- (7_1.west) -- (7_1.east) -- (8_1.west) -- (8_1.east) -- (9_1.west) -- (9_1.east) -- (10_2.west) -- (10_2.east) -- (11_2.west) -- (11_2.east) -- (12_3.west) -- (12_3.east) -- (13_4.west) -- (13_4.east);
\end{tikzpicture}

\section{Young's Rule}

In this section we show how Young scores can be computed in polynomial time by setting up integer programs whose LP relaxations have totally unimodular matrices. This means that vertices of the feasible regions have integer coordinates, so it is enough to solve the LPs.

\begin{theorem}\label{theorem_young} Given a \two-crossing profile $\calP$ with voter set $V$ and candidate set $C$, and a candidate $c \in C$, the weak and strong Young scores of $c$ can be computed in polynomial time.
\end{theorem}
\begin{proof}
We present the algorithm for the weak Young score, the strong case being analogous. We are hence interested in removing a minimal number of voters such that the majority margins satisfy $m_{c, c'} \geq 0$ for all $c' \in C \setminus \{c\}$. To do so, set up an integer program $P$ with variables $(x_v)_{v \in V} \in \{0, 1\}$, where $x_v$ is $1$ iff voter $v$ is kept; i.e.~it is not removed. Then, $m_{c. c'} \geq 0$ can be rewritten as $\sum_{v \in V}x_v\den{c \succ_v c'} \geq 0$; where by $\den{S} = 2[S] - 1$ we mean $1$ when $S$ holds and $-1$ otherwise. The goal is to maximize $\sum_{v \in V}{x_v}$. If we ignore the $x_v \in \{0, 1\}$ constraints and henceforth assume that voters are ordered in a \two-crossing fashion, then the constraints matrix of $P$ satisfies a property similar to the row-wise version (implicit from now on) of \CConeP:~it consists of values $\pm1$ and on each row ones form a circular interval. Unfortunately, not all such matrices are totally unimodular; e.g~the $2 \times 2$ matrix with $1$ on the main diagonal and $-1$ otherwise. This means that theorems regarding total unimodularity do not immediately apply.

To mitigate this difficulty, we need a slightly different approach:~we will successively check for each value $s = n, n - 1, \ldots, 0$ whether a solution keeping exactly $s$ of the $n$ voters exists. For a fixed $s$, we add the additional constraint that $\sum_{v \in V}x_v = s$. With the sum of the $x$'s fixed we can rewrite our main inequalities as follows:
\begin{gather*}
    \sum_{v \in V}x_v\den{c \succ_v c'} \geq 0 \text{ iff }
    \sum_{v \in V}x_v\left(2\iverson{c \succ_v c'} - 1\right) \geq 0 \text{ iff } \\ 
    2\sum_{v \in V}x_v\iverson{c \succ_v c'} \geq \sum_{v \in V} x_v \text{ iff }
    \sum_{v \in V}x_v\iverson{c \succ_v c'} \geq \ceil*{\frac{s}{2}}
\end{gather*}

\noindent Therefore, for a fixed $s$, our integer program $P_s$ consists of checking the feasibility of the following constraints:
\begin{align}
&\sum_{v \in V}x_v = s &&  \label{cons_1} \\
&\sum_{v \in V}x_v\iverson{c \succ_v c'} \geq \ceil*{\frac{s}{2}}, && \text{$c' \in C \setminus \{c\}$}  \label{cons_2} \\
&x_v \in \{0, 1\}, && \text{$v \in V$} \label{cons_3}
\end{align}

\noindent At this point, if we relax constraints (\ref{cons_3}) to $0 \leq x_v \leq 1$, then the matrix of the resulting linear program $L_s$ satisfies \CConeP. Unfortunately, this still does not guarantee total unimodularity; e.g.~the $3 \times 3$ matrix with $0$ on the main diagonal and $1$ otherwise has determinant $2 \notin \{-1, 0, 1\}$. 

In one last step, we slightly adjust our constraints to get a matrix which is \ConeP, and hence provably totally unimodular \cite{fulkerson_gross}. Note that constraints (\ref{cons_1}) and (\ref{cons_3}) do not make use of circularity, so only constraints (\ref{cons_2}) might need adjustments. Consider some constraint $\sum_{v \in V}x_v\iverson{c \succ_v c'} \geq \ceil*{\frac{s}{2}}$ needing adjustment; this is because $\iverson{c \succ_v c'}$ is of the form $11\ldots1100\ldots0011\ldots11$ as $v$ ranges over the voters. We can then write
\begin{gather*}
    \sum_{v \in V}x_v\iverson{c \succ_v c'} \geq \ceil*{\frac{s}{2}} \text{ iff } \sum_{v \in V}x_v\left(1 - \iverson{c' \succ_v c}\right) \geq \ceil*{\frac{s}{2}} \text{ iff } \\
     s - \sum_{v \in V}x_v\iverson{c' \succ_v c} \geq \ceil*{\frac{s}{2}} \text{ iff } \sum_{v \in V}x_v\iverson{c' \succ_v c} \leq \floor*{\frac{s}{2}}
\end{gather*}
\noindent For this equivalent condition, the coefficients $\iverson{c' \succ_v c}$ are of the form $00\ldots0011\ldots1100\ldots00$ as $v$ ranges over $V$. Therefore, if we replace each constraint which uses circularity as such, the matrix of the resulting LP will be \ConeP, and so totally unimodular. This means that the feasible region has integer vertices \cite{hoffman_kruskal},
so the IP and the LP are equifeasible.
Note that the replacements we made in (\ref{cons_2}) proved total unimodularity, but need not be executed in practice; i.e.~it is enough to check the feasibility of $L_s$. If we now use any polynomial time algorithm for checking feasibility, like \cite{karmarkar}, or any more sophisticated interior-point methods, the poly-time bound follows.
\end{proof}

In the following, we further refine our technique to get an algorithm of a more practical worst-case time complexity.

\begin{theorem}\label{theorem_improved_young}
Given a \two-crossing profile $\calP$ and $c \in C$, the strong/weak Young score of $c$ can be computed in time $O((n + m^2)n^{3/2}\log{n})$.
\end{theorem}

\begin{proof}
We start at the end of the proof of Theorem \ref{theorem_young}, from the LP whose matrix has the consecutive ones property. Our goal is to check the feasibility of the following constraints $C_s$:
\begin{align}
    & 0 \leq x_v \leq 1,                                                  & v \in V       \label{c_s_1} \\
    & \sum_{x \in V} x_v = s                                              &               \label{c_s_2} \\
    & \sum_{v \in V}x_v\iverson{c \succ_v c'} \geq \ceil*{\frac{s}{2}},  & (c, c') \in A \label{c_s_3} \\
    & \sum_{v \in V}x_v\iverson{c' \succ_v c} \leq \floor*{\frac{s}{2}},   & (c, c') \in B \label{c_s_4}
\end{align}

\noindent where $A, B$ are sets such that $A \cap B = \varnothing$, $A \cup B = \{(c, c') \in C^2 : c \neq c'\}$.

\noindent Recall that $V = [n]$. We will now give an alternative formulation of $C_s$ using $n + 1$ variables:~$S_0, \ldots, S_n$, constructed such that $S_i - S_{i - 1} = x_i$ for $i \in [n]$; i.e.~the prefix sums of $(x_i)_{i \in [n]}$. For these purposes, for each pair of distinct candidates $(c, c') \in C^2$ look at the corresponding row in \eqref{c_s_3}-\eqref{c_s_4} of the matrix of $C_s$ and let $\ell_{(c, c')}$ and $r_{(c, c')}$ be the indices of the leftmost/rightmost one on that row. Then, $C_s$ can be reformulated as follows, which we call $D_s$:
\begin{align}
    & 0 \leq S_i - S_{i - 1} \leq 1,                       & i \in [n] \label{d_s_1} \\
    & S_n - S_0 = s                                        &          \label{d_s_2} \\
    & S_{r_x} - S_{\ell_x - 1} \geq \ceil*{\frac{s}{2}},  & x \in A   \label{d_s_3} \\
    & S_{r_x} - S_{\ell_x - 1} \leq \floor*{\frac{s}{2}},   & x \in B   \label{d_s_4}
\end{align}
\noindent Constraint set $C_s$ is feasible iff constraint set $D_s$ is feasible, so it is enough to check whether $D_s$ is feasible. Now, observe that constraints in \eqref{d_s_1}--\eqref{d_s_4} can all be rewritten using constraints of the form ${S_i - S_j \leq \mathcal{C}}$, where $\mathcal{C}$ is integer and $i, j \in [0:n]$. This shows that $D_s$ is a so-called ``difference constraints system''. In \cite{cormen}, the standard approach for solving such systems is explained, which we now apply to our situation. Let $G$ be a weighted directed graph with $n + 2$ vertices $V_G = [0 : n] \cup \{*\}$ and edges $E_G$ as follows:~for all $i \in [0 : n]$ add an edge $* \xrightarrow{0} i$, and for all constraints $S_i - S_j \leq \mathcal{C}$ in $D_s$ add an edge $j \xrightarrow{\mathcal{C}} i$. As argued in \cite{cormen}, $D_s$ is satisfiable iff $G$ does not have a negative-weight cycle. Moreover, if $G$ has no negative-weight cycles, then $S = (\delta_*(0), \delta_*(1), \ldots, \delta_*(n))$ is a solution to $D_s$, where $\delta_*(i)$ is the distance between $*$ and $i$ in $G$. Checking for negative cycles and computing distances from $*$ can both be achieved in time $O\left(|V_G||E_G|\right) = O(n(n + m^2))$ using the Bellman-Ford algorithm, or faster alternatives, like SPFA, which will, on most practical instances, have far superior runtime. If we are interested in an algorithm with better worst-case theoretical guarantees, then the algorithm of \cite{faster_shortest_paths} solves the problem in $O\left(|E_G|\sqrt{|V_G|}\log{s}\right) = O((n + m^2)\sqrt{n}\log{n})$. Since this has to be run for every $s = n, \ldots, 0$, the overall complexity becomes $O((n + m^2)n^{3/2}\log{n})$. Note that, unlike Theorem \ref{theorem_young}, this result uses that recognition is tractable in $O(nm^2)$.
\end{proof}

\section{The Chamberlin--Courant Rule}

In this section we show that computing the least possible dissatisfaction and also a winning committee for the Chamberlin--Courant rule (CC in this section) can both be achieved in polynomial time under \two-crossing elections. 
Since the recognition problem can be solved in polynomial time by Theorem \ref{theorem_reco}, assume voters $V = [n]$ are ordered such that our profile is \two-crossing with respect to the identity permutation. For the following key lemma, we call a $k$-assignment function $w$ \emph{illegal} if there are two candidates $a \neq b$ and four voters $i_1 < i_2 < i_3 < i_4$ such that $w(i_1) = w(i_3) = a$ and $w(i_2) = w(i_4) = b$. We call $w$ \emph{legal} if it is not illegal.

\begin{lemma} \label{lemma_structure} For any instance $(\calP, \rho, k)$ of CC there exists a legal $k$-assignment $w_{opt}$ that is optimal for $\rho$ and $\calP$.
\end{lemma}
\begin{proof} Start with any optimal $k$-assignment $w$. If $w$ is legal, then we are done, otherwise let $a \neq b$ and $i_1 < i_2 < i_3 < i_4$ witnesses the illegality of $w$. Since $\calP$ is \two-crossing we know it can not be the case that $\iverson{a \succ_{i_1} b} = \iverson{a \succ_{i_3} b} \neq \iverson{b \succ_{i_2} a} = \iverson{b \succ_{i_4} a}$, as that would result in 3 crosses for the pair $(a, b)$. Therefore, for at least one of $i_1, i_2, i_3, i_4$ assignment $w$ assigns a candidate from $\{a, b\}$ which is less preferred than the other, so just exchanging $a$ for $b$, or vice-versa, for that voter results in an assignment $w'$ such that $\Phi_\rho(\calP, w') \leq \Phi_\rho(\calP, w)$. Since $w$ was optimal, it follows that $w'$ is also optimal. We can now replace $w$ by $w'$ and repeat the reasoning iteratively, until we eventually reach a legal optimal $k$-assignment $w_{opt}$. It remains to show that this process terminates, but this is easy to see since at each step we replace a voter's assigned candidate with one that is strictly more preferred by them. 
\end{proof}

\noindent To find a winning committee for CC it is enough to find an optimal $k$-assignment for $\rho$ and $\calP$.
Given Lemma \ref{lemma_structure}, we can limit our search to legal $k$-assignments. It turns out that legal $k$-assignments admit a second characterisation, as follows. The proof is straightforward, so we leave it to the reader.
\begin{proposition}\label{prop_tree} A $k$-assignment $w = w(1), w(2), \ldots, w(n)$ is legal iff for any candidates $a \neq b$ there are no $a$'s between the first and the last $b$ or no $b$'s between the first and last $a$.
\end{proposition}

\noindent Given this, any legal assignment $w$ induces a strict partial order relation $\rightarrow^*$ over $C$, where $a \rightarrow^* b$ for $a \neq b$ iff there are $b$'s between the first and the last $a$ in $w$. Note that this implicitly requires no $a$'s between the first and the last $b$, by Proposition \ref{prop_tree}. Observe that two candidates $a \neq b$ are incomparable with respect to $\rightarrow^*$ iff all $a$'s precede all $b$'s in $w$, or vice-versa. Relation $\rightarrow^*$ satisfies an additional property:~if $a \neq b$ are incomparable, then there can be no $c \notin \{a, b\}$ such that $a \rightarrow^* c$ and $b \rightarrow^* c$, since a structure of the form $ a\ldots c \ldots a \ldots b\ldots c \ldots b $ occurring in $w$ would be illegal for pairs $(a, c)$ and $(b, c)$. These being said, one can now observe that the covering relation $\rightarrow$ of $\rightarrow^*$; i.e.~the inclusion minimal relation with the same transitive closure; is a forest of directed trees.\footnote{We do not use this fact explicitly, but it helps gain intuition for the DP which follows.} Since our partial order $\rightarrow^*$ is finite, there exists at least one maximal element $c_{\mathit{rt}} \in C$, which we call a ``root'' element; i.e.~such that there is no $c \in C$ with $c \rightarrow^* c_{\mathit{rt}}$.
Let us now analyze $w$ with respect to $c_{\mathit{rt}}$. Since $c_{\mathit{rt}}$ is maximal, for all $c \in C \setminus \{c_{\mathit{rt}}\}$, no $c_{\mathit{rt}}$'s occur between the first and the last $c$ in $w$. In other words, if we consider the set $w^{-1}(c_{\mathit{rt}})$ of positions where $c_{\mathit{rt}}$ occurs in $w$, then all $c$'s occur between two consecutive $c_{\mathit{rt}}$'s, or before/after the first/last $c_{\mathit{rt}}$. Put differently, candidate $c_{\mathit{rt}}$ ``splits'' the range $1 \ldots n$ into buckets delimited by consecutive entries in $w^{-1}(c_{\mathit{rt}}) \cup\{0, n + 1\}$, such that any candidate $c \in C \setminus \{c_{\mathit{rt}}\}$ appears in at most one bucket. This reasoning can then be carried out recursively in each bucket as well.

The idea of our approach can now be stated intuitively:~we will proceed by dynamic programming, solving the main problem by trying out all possible values of $w^{-1}(c_{\mathit{rt}})$, and then invoking recursively on each resulting bucket. There are a number of burning issues at this point:~(i) there are exponentially many values of $w^{-1}(c_{\mathit{rt}})$ to try, (ii) there are exponentially many ways to distribute the number $k$ of representatives into the buckets
(iii) how do we ensure that no candidate is used in two distinct buckets, maybe further down the recursion, fact which might render the assignment illegal? Before addressing (i)-(ii) specifically, let us first give a preliminary version of the dynamic program, which will run in exponential time, and argue for its correctness in addressing (iii). Introduce $\dyp[\ell, r, t]$ for $1 \leq \ell \leq r \leq n$, $1 \leq t \leq k$ to mean the least total dissatisfaction attainable for voters in $[\ell : r]$ using a legal $t$-assignment restricted to those voters. For ease of presentation, we use the degenerate base-cases $\dyp[\ell, \ell - 1, -] = 0$ for $\ell \in [n + 1]$ and $\dyp[\ell, r, 0] = \infty$ for $1 \leq \ell \leq r \leq n$. Then, the recurrence relation is:%
\begin{gather}
     \dyp[\ell, r, t] = \min\left\{\sum_{i = 1}^{x} \rho(w_i, c_{\mathit{rt}})+ \sum_{i =  0}^{x}\dyp[w_i + 1, w_{i + 1} - 1, t_i] : \right. \nonumber \\
     c_{\mathit{rt}} \in C, 1 \leq x \leq r - \ell + 1, \ell \leq w_1 < \ldots < w_x \leq r, \nonumber \\ 
     t_0, \ldots, t_x \geq 0 \text{ and } t_0 + \ldots + t_x = t - 1 \bigg\}
     \label{eq_rec}
\end{gather}%
\noindent where we wrote $w_1, \ldots, w_x$ for the elements of $w^{-1}(c_{\mathit{rt}})$ and used the convention that $w_0 = \ell - 1$ and $w_{x + 1} = r + 1$. In other words:~we choose the root candidate $c_{\mathit{rt}}$, we choose $w^{-1}(c_{\mathit{rt}})$, and we choose a way to distribute the remaining $t - 1$ candidates among the buckets; each bucket is then solved ``recursively'' with the respective allowed committee size. For completeness, subproblems can be computed in increasing order of $t$, breaking ties arbitrarily. The cost of an optimal $k$-assignment $w_{opt}$ will be in $\dyp[1, n, k]$ at the end. To retrieve one such $w_{opt}$,  additional bookkeeping is required:~for each subproblem $(\ell, r, x)$ keep track of which values of $c_{\mathit{rt}}$, $x$, $w_1, \ldots, w_x$ and $t_0, \ldots, t_x$ were used to achieve the minimum, and then at the end trace those back to build $w_{opt}$.

As argued previously, not only does this DP require exponential time to compute, it is not even clearly correct, reason being that it can produce illegal $k$-assignments by reusing candidates, especially along different branches of the recursion. However, note that this comes at a price:~whenever a candidate is reused it is counted as a new candidate out of the maximum of $k$ allowed. Nevertheless, none of this is harmful:~consider the reconstructed optimal assignment $w_{opt}$ returned by the DP. This is a $k$-assignment since we only ever over-count candidates in the assembly, but never under-count. By construction, our recurrence clearly correctly considers all legal $k$-assignments, and will thus return a cost at most that of the least dissatisfaction legal $k$-assignment. Since by Lemma \ref{lemma_structure} no illegal $k$-assignment can be better than the best legal one, it follows that the returned assignment $w_{opt}$ is optimal, regardless of whether it is legal or not.

This proof approach can be summarized as follows, and is also implicit in the works of \cite{constantinescu_elkind_2021}:~(1) show optimal solutions with a rigid structure exist; (2) set up a DP which looks only for solutions of this form; (3) note that the DP occasionally produces non-conforming solutions, but without affecting global correctness.

Having addressed concern (iii), we now need to speed up the DP --- we need a faster way to compute \eqref{eq_rec}. To achieve this, introduce an auxiliary second dynamic program, $\dyp_2[\ell, r, t, c_{\mathit{rt}}]$, with the same semantics as $\dyp$, but enforcing a certain value for the root $c_{\mathit{rt}}$. It is straightforward to see that 
\begin{equation} \label{eq_easy}
\dyp[\ell, r, t] = \min_{c_{\mathit{rt}} \in C}\dyp_2[\ell, r, t, c_{\mathit{rt}}]
\end{equation}
so it remains to show how $\dyp_2$ can be computed in polynomial time, which we do with the following.

\begin{lemma} $\dyp_2$ satisfies the following recurrence relation:
\begin{gather}
\dyp_2[\ell, r, t, c_{\mathit{rt}}] = 
\min\big\{ \dyp[\ell, w_1 - 1, t_0] + \rho(w_1, c_{\mathit{rt}}) + \nonumber \\  
        \min\left\{\dyp[w_1 + 1, r, t - t_0 - 1],
              \dyp_2[w_1 + 1, r, t - t_0, c_{\mathit{rt}}]\right\} : \nonumber \\ w_1 \in [\ell, r], 0 \leq t_0 \leq t - 1 \big\} \label{eq_better_rec}
\end{gather}
\end{lemma}

\begin{proof}
The key stands in noting that it makes little sense to try out all values of $w^{-1}(c_{\mathit{rt}}) = \{w_1, \ldots\}$ at once, so let us only go through all values of $w_1 \in [\ell : r]$ and $t_0 \in [0 : t - 1]$ and observe the optimal substructure:~for the left part, i.e.~range $[\ell : w_1 - 1]$, we use the optimal unrestricted solution with at most $t_0$ candidates, i.e.~$\dyp[\ell, w_1 - 1, t_0]$, for position $w_1$ we take the fixed cost $\rho(w_1, c_{\mathit{rt}})$, while for the right part, i.e.~range $[w_1 + 1, r]$, we have more choice. Namely, there are two possibilities: $w_1$ was the only element in $w^{-1}(c_{\mathit{rt}})$, in which case we need to take the best unrestricted solution with at most $t - t_0 - 1$ candidates, the $-1$ accounting for candidate $c_{\mathit{rt}}$, i.e.~$\dyp[w_1 + 1, r, t - t_0 - 1]$, or $|w^{-1}(c_{\mathit{rt}})| > 1$, in which case there is a $w_2$ (and potentially more occurrences of $c_{\mathit{rt}}$) to be placed; doing so optimally is precisely the definition of $\dyp_2[w_1 + 1, r, t - t_0, c_{\mathit{rt}}]$.\footnote{Note no $-1$ here because $c_{\mathit{rt}}$ is still actively ``in use'' at this point.} Altogether, we get \eqref{eq_better_rec}. 
\end{proof}

We now outline our main result. Essentially, one can now compute subproblems in increasing order of $r - \ell$, the optimal dissatisfaction being in $\dyp[1, n, k]$ at the end.

\begin{theorem}\label{CC_is_poly_time} Given an instance $(\calP, \rho, k)$ of CC where $\calP$ is \two-crossing, the optimal dissatisfaction and some winning committee can be computed in polynomial time.
\end{theorem}

\begin{proof}
Without loss of generality, assume $\calP$ is \two-crossing with respect to the identity permutation (otherwise, we can use our recognition algorithm to find a witnessing permutation). We compute $\dyp[\ell, r, t]$ and $\dyp_2[\ell, r, t, c_r]$ for all $1 \leq \ell \leq r \leq n$, $1 \leq t \leq k$ and $c_r \in C$ using \eqref{eq_easy} and \eqref{eq_better_rec} in increasing order of $r - \ell$, breaking ties arbitrarily, but only computing $\dyp[\ell, r, t]$ once $\dyp_2[\ell, r, t, c_r]$ has been computed for all $c_r \in C$. For the time complexity, there are $O(n^2k)$ subproblems of the form $\dyp[-, -, -]$, each taking time $O(m)$ to compute, and there are $O(n^2km)$ subproblems of the form $\dyp_2[-, -, -, -]$, each taking time $O(nk)$ to compute. Altogether, this amounts to $O(n^3k^2m)$. So far this is enough to give us the optimal dissatisfaction $\dyp[1, n, k]$. To also retrieve a winning committee, one needs additional bookkeeping for each subproblem, similar to the one described previously; this does not change the overall time complexity.
\end{proof}

We note that for ``egalitarian'' Chamberlin--Courant, where one is instead interested in minimizing the dissatisfaction of the most misrepresented voter, simply replacing `\texttt{+}' with \texttt{max} in the DPs preserves correctness.

\section{Conclusions and Future Work}

We investigated \two-crossing elections, giving an efficient recognition algorithm. Axiomatically, we showed that \two-crossing elections are no different from general elections for any voting rule operating on the (weighted) majority tournament, such as Kemeny and Slater. Subsequently, we considered Young and the Chamberlin--Courant rule. In both cases we gave polynomial algorithms which are applicable in practice. We leave open whether faster solutions exist.

So far, \two-crossing elections have not received the attention they deserve in the social choice literature. It would be interesting to run larger scale experiments to determine to what extent real election data obeys the \two-crossing model, such as by checking for the property in PrefLib \cite{mattei_walsh_preflib}.
Algorithmically, we leave open whether Dodgson winners can be computed in polynomial time under \two-crossing elections. $k$-crossing preferences for $k > 2$ are strictly more expressive than their \two-crossing counterparts, but this comes with an increase in computational complexity; e.g.~computing a winning committee for the Chamberlin--Courant rule, while being in P for \two-crossing elections, becomes NP-hard under three-crossing elections. It would be interesting to establish such tight bounds for other voting rules; e.g.~Young's; and, perhaps most importantly, to establish the hardness of recognizing $k$-crossing elections for $k > 2$. Axiomatically, it would be interesting to determine whether \two-crossing elections admit a forbidden structure characterization, similarly to single-crossing elections \cite{single_crossing_characterization}.

\section*{Acknowledgements}

We thank Edith Elkind for the many useful discussions about two-crossing elections. We thank Costin Andrei Oncescu for the main proof idea of Theorem \ref{theorem_improved_young}.



\bibliographystyle{named}
{\small
\bibliography{ijcai22}}

\end{document}